\documentclass{eptcs}
\providecommand{\event}{Refine 2011}

\usepackage{amsmath}
\usepackage{amssymb}
\usepackage{eventB}
\usepackage{zed-csp}
\usepackage{color}
\usepackage{transparent}
\usepackage{graphicx}

\newcommand{\Bmachine}{\;\textbf{machine}}
\newcommand{\Bvariables}{\;\textbf{variables}}

\newcommand{\Binvariant}{\;\textbf{invariant}}
\newcommand{\Bvariant}{\;\textbf{variant}}

\newcommand{\Bevents}{\;\textbf{events}}

\def\Bwhen{\;\mbox{\bf when}\;}
\def\Bthen{\;\mbox{\bf then}\;}
\def\Bend{\;\mbox{\bf end}}
\def\Brefines{\;\mbox{\bf refines}}
\def\Bstatus{\;\mbox{\bf status}}

\def\comp{;}

\newcommand{\traces}{traces}
\newcommand{\divergences}{divergences}

\newcommand{\brefinedby}{\preccurlyeq}

\def\emptyset{\{ \}}

\newtheorem{theorem}{Theorem}[section]
\newtheorem{lemma}[theorem]{Lemma}

\newtheorem{definition}[theorem]{Definition}

\newtheorem{proof}[theorem]{Proof}

\title{A CSP account of Event-B refinement}
\author{Steve Schneider 
\institute{Department of Computing, University of Surrey} 
\email{S.Schneider@surrey.ac.uk}
\and
Helen Treharne \institute{Department of Computing, University of Surrey} 
\email{H.Treharne@surrey.ac.uk}
\and
Heike Wehrheim \institute{Department of Computer Science, University of Paderborn}
\email{wehrheim@uni-paderborn.de}
}
\def\titlerunning{A CSP account of Event-B refinement}
\def\authorrunning{S.Schneider, H. Treharne \& H. Wehrheim}  

\begin{document}

\maketitle


\begin{abstract}
Event-B provides a flexible framework for stepwise system development via refinement. The framework supports steps for (a) refining events (one-by-one), (b) splitting events (one-by-many), and (c) introducing new events. In  each of the steps events can moreover possibly be anticipated or convergent. All such steps are accompanied with precise proof obligations. Still, it remains unclear what the exact relationship - in terms of a {\em behaviour-oriented semantics} - between an Event-B machine and its refinement is. In this paper, we give a CSP account of Event-B refinement, with a treatment for the first time of splitting events and of anticipated events. To this end, we define a CSP semantics for Event-B and show how the different forms of Event-B refinement can be captured as CSP refinement.  
\end{abstract}

\section{Introduction}

Event-B \cite{Abrial09} provides a framework for system development through stepwise refinement.  Individual refinement steps are verified with respect to their proof obligations, and the transitivity of refinement ensures that the final system description is a refinement of the initial one.  The refinement process allows new events to be introduced through the refinement process, in order to provide the more concrete implementation details necessary as refinement proceeds. 

The framework allows for a great deal of flexibility as to cover a broad range of system developments. The recent book \cite{Abrial09} comprising case studies from rather diverse areas shows that this goal is actually met. The flexibility is a result of the different ways of dealing with events during refinement. At each step existing events of an Event-B machine need to be refined. This can be achieved by (a) simply keeping the event as is, (b) refining it into another event, possibly because of a change of the state variables, or (c) splitting it into several events\footnote{A fourth option is merging of events which we do not consider here.}. Furthermore, every refinement step allows for the introduction of new events. To help reasoning about divergence, events are in addition classified as ordinary, {\em anticipated} or {\em convergent}. Anticipated and convergent events both introduce new details into the machine specification. Convergent events must not be executed forever, while for anticipated events this condition is deferred to later refinement steps. All of these steps come with precise proof obligations; appropriate tool support helps in discharging these \cite{AbrialBHV08,DBLP:journals/sttt/AbrialBHHMV10}.   Event-B is essentially a {\em state-based} specification technique, and proof obligations therefore reason about predicates on states.   

Like Event-B, CSP comes with a notion of refinement.  In order to understand their relationship, these two refinement concepts need to be set in a single framework. Both formalisms moreover support a variety of different forms of refinement: Event-B by means of several proof obligations related to refinement, out of which the system designer chooses an appropriate set; CSP by means of its different semantic domains of traces, failures and divergences.   The aim of this paper is to give a precise account of Event-B refinement in terms of CSP's behaviour-oriented  process refinement.  This will also provide the underlying results that support refinement in the combined formalism Event-B$\|$CSP. 
Our work is thus in line with previous studies relating state-based with behaviour-oriented refinement (see e.g.\ \cite{BolDav02,Derrick02c,BoitenD09}). It turns out that CSP supports an approach to refinement consistent with that of Event-B. It faithfully reflects all of Event-B's possibilities for refinement, including splitting events and new events. It moreover also deals with the Event-B approach of  anticipated events as a means to defer consideration of divergence-freedom.  Our results involves support for individual refinement steps as well as for the resulting refinement chain. 


\medskip

\noindent The paper is structured as follows. The next section introduces the necessary background on Event-B and CSP.  Section 3 gives the CSP semantics for Event-B based on weakest preconditions. In Section 4 we precisely fix the notion of refinement used in this paper, both for CSP and for Event-B, and Section 5 will then set these definitions in relation. It turns out that the appropriate refinement concept of CSP in this combination with Event-B is infinite-traces-divergences refinement. The last section concludes. 

\section{Background} 

We start with a short introduction to CSP and Event-B. For more detailed information see \cite{schneider99} and \cite{Abrial09} respectively.

\subsection{CSP}

CSP, Communicating Sequential Processes, introduced by Hoare \cite{hoare85} is a formal specification language aiming at the description of communicating processes. A process is characterised by the events it can engage in and their ordering. Events will in the following be denoted by $a_1, a_2, \ldots$ or $evt0, evt1, \ldots$. Process expressions are built out of events using a number of composition operators. In this paper, we will make use of just three of them:  {\em interleaving} ($P_1 ||| P_2$), executing two processes in parallel without any synchronisation; {\em hiding} ($P \hide N$), making a set $N$ of events internal; and {\em renaming} ($f(P)$ and $f^{-1}(P)$), changing the names of events according to a renaming function $f$. If $f$ is a non-injective function, $f^{-1}(P)$ will offer a choice of events $b$ such that $f(b) = a$ whenever $P$§ offers event $a$. 

Every CSP process $P$ has an alphabet $\alpha P$.  Its semantics is given using the Failures/Divergences/Infinite Traces semantic model for CSP.  This is presented as ${\cal U}$ in \cite{roscoe:tpc} or FDI in \cite{schneider99}.
The semantics of a process can be understood in terms of four sets, $T, F, D, I$, which are respectively the traces, failures, divergences, and infinite traces of $P$.  These are understood as observations of possible executions of the process $P$, in terms of the events from $\alpha P$ that it can engage in.

Traces are finite sequences of events from $P$'s alphabet: $tr \in \alpha P^*$.  The set $traces(P)$ represents the possible finite sequences of events that $P$ can perform.  
Failures will not be considered in this paper and are therefore not explained here.

Divergences are finite sequences of events on which the process might diverge: perform an infinite sequence of internal events (such as an infinite loop) at some point during or at the end of the sequence.  The set $divergences(P)$ is the set of all possible divergences for $P$.  
Infinite traces $u \in \alpha P^{\omega}$ are infinite sequences of events.  The set $infinites(P)$ is the set of infinite traces that $P$ can exhibit.  For technical reasons it also contains those infinite traces which have some prefix which is a divergence.

\begin{definition}
A process $P$ is {\em divergence-free} if $divergences(P) = \{ \}$.
\end{definition}

\noindent We use $tr$ to refer to finite traces.  These can also be written explicitly as $\langle a_1, a_2, \ldots, a_n \rangle$.  The empty trace is $\langle \rangle$, concatenation of traces is written as $tr_1 \cat tr_2$.  We use $u$ to refer to infinite traces.  Given a set of events $A$, the {\em projections} $tr \project A$ and $u \project A$ are the traces restricted to only those events in $A$.  Note that $u \project A$ might be finite, if only finitely many $A$ events appear in $u$.  Conversely, $tr \hide A$ and $u \hide A$ are those traces with the events in $A$ removed.  The length operator $\# tr$ and $\# u$ gives the length of the trace it is applied to. As a first observation, we get the following.

\begin{lemma} \label{lem:infdivfree}
If $P$ is divergence-free, and for any infinite trace $u$ of $P$ we have $\#(u \hide A) = \infty$, then $P \hide A$ is divergence-free.
\end{lemma}
\begin{proof}
Follows immediately from the semantics of the hiding operator. 
\end{proof}

\noindent Later, we furthermore use {\em specifications} on traces or, more generally, on CSP processes. Specifications are given in terms of predicates.  If $S$ is a predicate on a particular semantic element, then we write $P \sat S$ to denote that all relevant elements in the semantics of $P$ meet the predicate $S$.  For example, if $S(u)$ is a predicate on infinite traces, then $P \sat S(u)$ is equivalent to $\forall u \in infinites(P) \,.\, S(u)$.

\subsection{Event-B}

Event-B~\cite{Abrial09,rodin:d7} is a state-based specification formalism based
on set theory. Here we describe the basic
parts of an Event-B machine required for this paper; a full description of the formalism can be found in \cite{Abrial09}.  

A machine specification usually defines a list of variables, given as $v$.  Event-B also in general allows sets $s$ and constants $c$.  However, for our purposes the treatment of elements such as sets and constants are independent of the results of this paper, and so we will not include them here. However, they can be directly incorporated without affecting our results.

There are many clauses that may appear in Event-B machines, and we concentrate on those clauses concerned with the state. 
We will therefore describe a machine $M_0$ with a list of state variables $v$, a state invariant $I(v)$, and a set of events $evt0, \ldots$ to update the state (see left of Fig.\ref{fig:machine}).  Initialisation is a special event $init0$. 

\begin{figure}
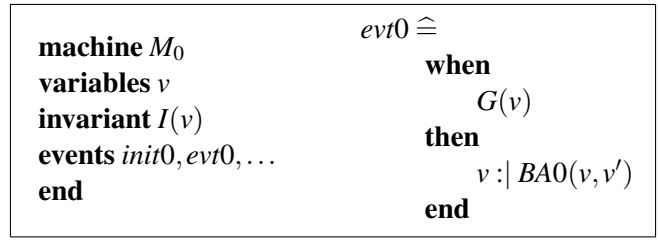
 
\begin{center} 
\framebox{$\begin{array}{l}
\Bmachine \; M_0 \\
\Bvariables \; v \\
\Binvariant \; I(v) \\
\Bevents \; init0, evt0, \ldots \\
\Bend \\ 
\end{array} \qquad 
\begin{array}{l}
evt0 \defs \\
\qquad \Bwhen \\
\qquad \qquad G (v) \\
\qquad \Bthen \\
\qquad \qquad v :| BA0(v,v') \\
\qquad \Bend \\ 
\end{array}$}
\end{center}
\caption{Template of an Event-B machine and an event.}
\label{fig:machine}
\end{figure} 

A machine $M_0$ will have various proof obligations on it.  These include consistency obligations, that events preserve the invariant.  They can also include (optional) deadlock-freeness obligations: that at least one event guard is always true.

Central to an Event-B description is
the definition of the events, each consisting of a {\em guard} $G(v)$ over the variables, and a {\em body}, usually written as an assignment $S$ on the variables. 
  The body defines a {\em before-after predicate} $BA(v, v')$ 
describing changes of variables upon event execution, in terms of the relationship between the variable values before ($v$) and after ($v'$).  The body can also be written as $v :| BA(v,v')$, whose execution assigns to $v$ any value $v'$ which makes the predicate $BA(v,v')$ true (see right of Fig.~\ref{fig:machine}). 

\section{CSP semantics for Event-B machine}

Event-B machines are particular instances of action systems, so Morgan's CSP semantics for action systems \cite{Morgan90} allows traces, failures, and divergences to be defined for Event-B machines, in terms of the sequences of events that they can and cannot engage in.  Butler's extension to handle unbounded nondeterminism \cite{butlerphd} defines the infinite traces for action systems. These together give a way of considering Event-B machines as CSP processes, and treating them within the CSP semantic framework.  In this paper we use the infinite traces model in order to give a proper treatment of divergence under hiding.  This is required to establish our main result concerning divergence-freedom under hiding of new events.  Consideration of finite traces alone is not sufficient for this result.

Note that the notion of {\em traces} for machines is different to that presented in \cite{Abrial09}, where traces are considered as sequences of {\em states} rather than our treatment of traces as sequences of {\em events}.


The CSP semantics is based on the weakest precondition semantics of events. Let $S$ be a statement (of an event). Then $[S]R$ denotes the weakest precondition for statement $S$ to establish postcondition $R$. Weakest preconditions for events of the form ``$\Bwhen\; G(v)\; \Bthen\; S(v)\; \Bend$'' are given by considering them as guarded commands:
\begin{eqnarray*}
[\Bwhen\; G(v)\; \Bthen\; S(v)\; \Bend]P & = & G(v) \implies [S(v)]P
\end{eqnarray*}
Events in the general form ``$\Bwhen\; G(v)\; \Bthen\; v :| BA(v,v')\; \Bend$'' have a weakest precondition semantics as follows:
\begin{eqnarray*}
[\Bwhen\; G(v)\; \Bthen\; v :| BA(v,v')\; \Bend]P & = & G(v) \implies \forall x . (BA(v,x) \implies P[x/v])
\end{eqnarray*}
Observe that for the case $P = true$ we have
\begin{eqnarray*}
[\Bwhen\; G(v)\; \Bthen\; v :| BA(v,v')\; \Bend]true & = & true
\end{eqnarray*}

\noindent Based on the weakest precondition, we can define the traces, divergences and infinite traces of an Event-B machine\footnote{Failures can be defined as well but are omitted since they are not needed for our approach.}. 

\begin{description}
\item[Traces] The traces of a machine $M$ are those sequences of events $tr = \trace{a_1,\ldots,a_n}$ which are possible for $M$ (after initialisation $init$): those that do not establish {\em false}:
\begin{eqnarray*}
traces(M) & = & \{ tr \mid \neg[init;  \negthickspace \negthinspace tr]false \}
\end{eqnarray*}
Here, the weakest precondition on a sequence of events is the weakest precondition of the sequential composition of those events: $[\trace{a_1,\ldots,a_n}]P$ is given as $[a_1 ; \ldots\, ; a_n]P = [a_1](\ldots ([a_n]P)\ldots)$.
\item[Divergences]
A sequence of events $tr$ is a divergence if the sequence of events is not guaranteed to terminate, i.e. $\neg[init;tr]true$.  Thus
\begin{eqnarray*}
divergences(M) & = & \{ tr \mid \neg[init;  \negthickspace \negthinspace tr]true \}
\end{eqnarray*}
Note that any Event-B machine $M$ with events of the form $evt$ given above is divergence-free.  This is because $[evt]true = true$ for such events (and for $init$), and so $[init;tr]true = true$.  Thus no potential divergence $tr$ meets the condition $\neg[init;tr]true$.  
\item[Infinite Traces]
The technical definition of infinite traces is given in \cite{butlerphd}, in terms of least fixed points of predicate transformers on infinite vectors of predicates.  Informally, an infinite sequence of events $u = \langle u_0, u_1, \ldots \rangle$ is an infinite trace of $M$ if there is an infinite sequence of predicates $P_i$ such that $\neg [init](\neg P_0)$ (i.e. some execution of $init$ reaches a state where $P_0$ holds), and $P_i \implies \neg[u_i](\neg P_{i+1})$ for each $i$ (i.e. if $P_i$ holds then some execution of $u_i$ can reach a state where $P_{i+1}$ holds).  
\begin{eqnarray*}
infinites(M) & = & \{ u \mid \mbox{there is a sequence}  \langle P_i \rangle_{i \in \nat} \;.\; \begin{array}[t]{l}
\neg[init](\neg P_0) \land \\
\mbox{for all}\;  i \;.\; P_i \implies \neg[u_i](\neg P_{i+1}) \; \}
\end{array}
\end{eqnarray*}
\end{description}

\noindent These definitions give the CSP Traces/Divergences/Infinite Traces semantics of Event-B machines in terms of the weakest precondition semantics of events.

\section{Refinement}

In this paper, we intend to give a CSP account of Event-B refinement. The previous section provides us with a technique for relating Event-B machines to the semantic domain of CSP processes. Next, we will briefly rephrase the refinement concepts in CSP and Event-B before explaining Event-B refinement in terms of CSP refinement.  

\subsection{CSP refinement}

Based on the semantic domains of traces, failures, divergences and infinite traces, different forms of refinement can be given for CSP. The basic idea underlying these concepts is - however - always the same: the refining process should not exhibit a behaviour which was not possible in the refined process. The different semantic domains then supply us with different forms of ``behaviour''. In this paper we will use the following refinement relation, based on traces and divergences: 
\begin{eqnarray*}
P \refinedby_{TDI} Q & \defs & 
\hspace*{-2mm}
\begin{array}[t]{l}
\ traces(Q) \subseteq traces(P) \\
{} \land divergences(Q) \subseteq divergences(P) \\
{} \land infinites(Q) \subseteq infinites(P) 
\end{array} 
\end{eqnarray*}

\noindent Refinement in Event-B also allows for the possibility of introducing new events. To capture this aspect in CSP, we need a way of incorporating this into process refinement. As a first idea, we could {\em hide} the new events in the refining process. This potentially introduces divergences, namely, when there is an infinite sequence of new events in the infinite traces. In order to separate out consideration of divergence from reasoning about traces, we will use $P \interleave RUN_N$ as a lazy abstraction operator instead. $RUN_N$ defines a divergence free process capable of executing any order of events from the set $N$. This will enable us to characterise Event-B refinement introducing new events in CSP terms.  The following lemma gives the relationship between refinement involving interleaving, and refinement involving hiding.

\begin{lemma} \label{lem:abs1}
If $P_0 \interleave RUN_N \refinedby_{TDI} P_1$ and $N \inter \alpha P_0 = \{ \}$ and $P_1 \hide N$ is divergence-free, then $P_0 \refinedby_{TDI} P_1 \hide N$.
\end{lemma}

\noindent {\bf Proof:} Assume  that (1) $P_0 \interleave RUN_N \refinedby_{TDI} P_1$, (2) $N \inter \alpha P_0 = \{ \}$  and (3) $P_1 \hide N$ is divergence-free. We need to show that the (finite and infinite) traces as well as divergences of $P_1 \hide N$ are contained in those of $P_0$. 
\begin{description}
\item[Traces] Let $tr \in traces(P_1 \hide N)$. By semantics of hiding there is some $tr' \in traces(P_1)$ s.t.\ $tr' \hide N = tr$. By (1) $tr' \in traces(P_0 \interleave RUN_N)$. By (2) and the semantics of $\interleave$ we get $tr' \hide N \in traces(P_0)$ and thus $tr \in traces(P_0)$. 
\item[Divergences] By (3) $divergences(P_1 \hide N) = \emptyset$, thus nothing to be proven here.
\item[Infinites] Let $u \in infinites(P_1 \hide N)$. By the semantics of hiding there is some $u' \in infinites(P_1)$ such that $u' \hide N = u$ and $\# (u' \hide N) = \infty$. By (1) $u' \in infinites(P_0 \interleave RUN_N)$ and by (2) and semantics of interleave we get $u' \hide N = u \in infinites(P_0)$.
\end{description} \hfill $\Box$

\subsection{Event-B refinement} 

In Event-B, the (intended) refinement relationship between machines is directly written into the machine definitions. As a consequence of writing a refining machine, a number of proof obligations come up.  Here, we assume a machine and its refinement to take the following form:

\smallskip 
\begin{minipage}{0.45\textwidth}
$
\begin{array}{l}
\Bmachine \; M_0 \\
\Bvariables \; v \\
\Binvariant \; I(v) \\
\Bevents \; init0, evt0, \ldots \\
\Bend \\ \\ \\
\end{array}
$
\end{minipage}
\begin{minipage}{0.45\textwidth}
$
\begin{array}{l}
\Bmachine \; M_1 \\
\Brefines \; M_0 \\
\Bvariables \; w \\
\Binvariant \; J(v,w) \\
\Bevents \; init1, evt1, \ldots \\
\Bvariant \; V(w) \\
\Bend
\end{array}
$
\end{minipage}

\medskip
\noindent The machine $M_0$ is actually refined by machine $M_1$, written $M_0 \brefinedby M_1$, if the given {\em linking invariant} $J$ on the variables of the two machines is established by their initialisations, and  preserved by all events, in the sense that any event of $M_1$ can be matched by an event of $M_0$ (or $skip$ for newly introduced events) to maintain $J$.  This is the standard notion of downwards simulation data refinement \cite{Derrick01a}. We next look at this in more detail, and in particular give the proof obligations associated to these conditions.

\begin{figure}
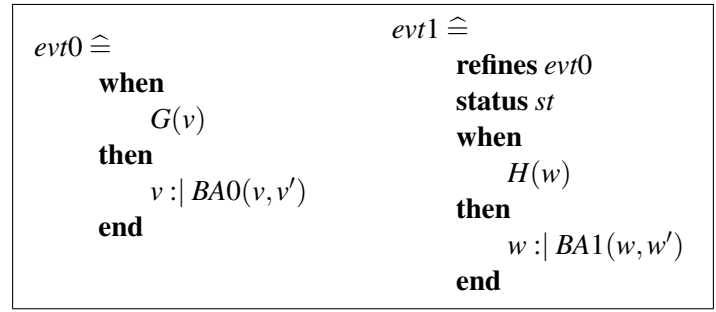

  \begin{center}
    \framebox{$
\begin{array}{l}
evt0 \defs \\
\qquad \Bwhen \\
\qquad \qquad G (v) \\
\qquad \Bthen \\
\qquad \qquad v :| BA0(v,v') \\
\qquad \Bend \\ \\ 
\end{array}
\qquad 
\begin{array}{l}
evt1  \defs \\
\qquad \Brefines \; evt0\\
\qquad \Bstatus \; st\\
\qquad \Bwhen \\
\qquad \qquad H(w) \\
\qquad \Bthen \\
\qquad \qquad w :| BA1(w,w') \\
\qquad \Bend
\end{array}
$
}
\end{center}
\caption{An event and its refinement}
\label{fig:events}
\end{figure}

First of all, we need to look at events again. Figure \ref{fig:events} gives the shape of an event and its refinement. We see that an event in the refinement now also gets a {\em status}. The status can be ordinary (also called {\em remaining}), or {\em anticipated} or {\em convergent}. Convergent events are those which must not be executed forever, and anticipated events are those that will be made convergent at some later refinement step.  New events must either have status anticipated or convergent. Both of these introduce further proof obligations: to prevent execution ``forever'' the refining machine has to give a variant $V$ (see above in $M_1$), and $V$ has to be decreased by every convergent event and must not be increased by anticipated events. 

We now describe each of the proof obligations in turn.  We have simplified them from their form in \cite{rodin:d7} by removing explicit references to sets and constants.  Alternative forms of these proof obligations are given in \cite[Section 5.2: Proof Obligation Rules]{Abrial09}.  

\begin{description}
\item[FIS\_REF: Feasibility]
Feasibility of an event is the property that, if the event is enabled (i.e. the guard is true), then there is some after-state.  In other words, the body of the event will not block when the event is enabled.

The rule for feasibility of a concrete event is:  
\begin{center}
$\Obligation
{
I (v) \land J (v, w) \land H (w)
}
{
\exists w' . BA1(w, w') 
}
{FIS\_REF}$ 
\end{center}

\item[GRD\_REF: Guard Strengthening]

This requires that when a concrete event is enabled, then so is the abstract one.  The rule is:
\begin{center}
$\Obligation
{
I (v) \land J (v, w) \land  H (w) 
}
{
G(v)
}
{GRD\_REF}$ 
\end{center}

\item[INV\_REF: Simulation]
This ensures that the occurrence of events in the concrete machine can be matched in the abstract one (including the initialization event).  New events are treated as refinements of $skip$.  The rule is:
\begin{center}
$\Obligation
{
I (v) \land J (v, w) \land H (w) \land BA1(w,w') 
}
{
\exists v' . (BA0(v,v') \land J(v',w'))
}
{INV\_REF}$ 
\end{center}
\end{description}

\noindent Event-B also allows a variety of further proof obligations for refinement, depending on what is appropriate for the application.  The two parts of the variant rule WFD\_REF below must hold respectively for all convergent and anticipated events, including all newly-introduced events.  

\begin{description}
\item[WFD\_REF: Variant]
This rule ensures that the proposed variant $V$ satisfies the appropriate properties: that it is a natural number, that  it decreases on occurrence of any convergent event, and that it does not increase on occurrence of any anticipated event:

\begin{center}
$\Obligation
{
I (v) \land J(v,w) \land H(w) \land BA1(w,w') 
}
{
V(w) \in \nat \land V(w') < V(w)
}
{\begin{tabular}{ll}WFD\_REF \\ (convergent event) \end{tabular}}$ 
\end{center}

\begin{center}
$\Obligation
{
I (v) \land J(v,w) \land H(w) \land BA1(w,w') 
}
{
V(w) \in \nat \land V(w') \le V(w)
}
{\begin{tabular}{ll}WFD\_REF \\ (anticipated event) \end{tabular}}$ 
\end{center}

\end{description}

\medskip
\noindent We will use the refinement relation $M_0 \brefinedby M_1$ to mean that the four proof obligations $FIS\_REF$, $GRD\_REF$, $INV\_REF$, and $WFD\_REF$ hold between abstract machine $M_0$ and concrete machine $M_1$.  

\section{Event-B refinement as CSP refinement} 

With these definitions in place, we can now look at our main issue, the characterisation of Event-B refinement via CSP refinement. Here, we in particular need to look at the different forms of events in Event-B during refinement.  
Events can have status convergent or anticipated, or might have no status.  This partitions the set of events of $M$ into three sets:  anticipated $A$, convergent $C$, and remaining events $R$ (neither anticipated nor convergent).  The alphabet of $M$, the set of all possible events, is thus given by $\alpha M = A \union C \union R$. In the CSP refinement, these will take different roles. 

Now consider an Event-B Machine $M_0$ and its refinement $M_1$: $M_0 \preccurlyeq M_1$.  
The machine $M_0$ has anticipated events $A_0$, convergent events $C_0$, and remaining events $R_0$, and $M_1$ similarly has event sets $A_1$, $C_1$, and  $R_1$.   Each event $ev_1$ in $M_1$ either refines a single event $ev_0$ in $M_0$ (indicated by the clause `refines $ev_0$' in the description of $ev_1$) or does not refine any event of $M_0$.  The set of new events $N_1$ is those events which are not refinements of events in $M_0$.

$M_0 \preccurlyeq M_1$ thus induces a partial surjective function $f_1 : \alpha M_1 \psurj \alpha M_0 $ where  $f_1(ev_1) = ev_0$  $\Leftrightarrow$ $ev_1\; \mbox{refines}\; ev_0$.   Observe that $\alpha M_1$ is partitioned by $f_1^{-1}(\alpha M_0)$ and $N_1$. 
The rules for refinement between events in Event-B impose restrictions on these sets: 
\begin{enumerate}
\item each event of $M_0$ is refined by at least one event of $M_1$;
\item each new event in $M_1$ is either anticipated or convergent;
\item each event in $M_1$ which refines an anticipated event of $M_0$ is itself either convergent or anticipated;
\item refinements of convergent or remaining events of $M_0$ are remaining in $M_1$, i.e. they are not given a status. 
\end{enumerate}

\noindent The conditions imposed by the rules are formalised as follows:
\begin{enumerate}
\item $ran(f_1) = A_0 \union C_0 \union R_0$;
\item $N_1 \subseteq A_1 \union C_1$;
\item $f_1^{-1}(A_0) \subseteq A_1 \union C_1$;
\item $f_1^{-1}(C_0 \union R_0) = f_1^{-1}(C_0) \union f_1^{-1}(R_0) = R_1$. 
\end{enumerate}

\begin{figure}[t]
\begin{center}
\def\svgwidth{0.5\textwidth}
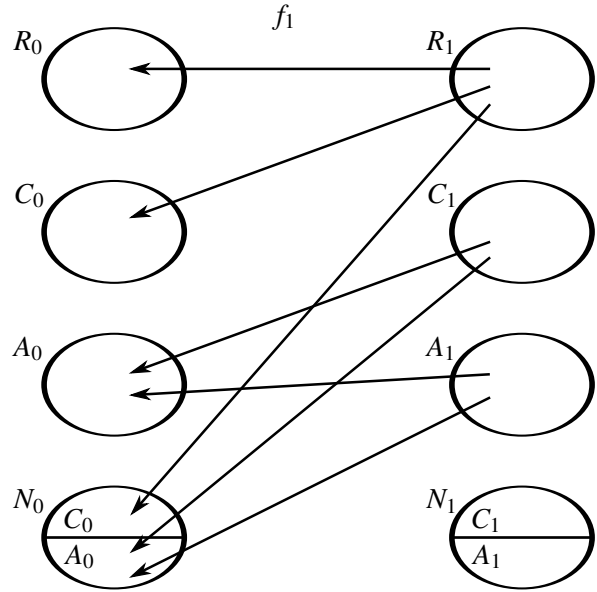
\end{center}
\caption{Relationship between events in a refinement step: $f_1$ maps events in $M_1$ to events in $M_0$ that they refine.} \label{fig:f}
\end{figure}

\noindent These relationships between the classes of events are illustrated in Figure~\ref{fig:f}. 

\subsection{New events}

For the new events arising in the refinement, we can use the lazy abstraction operator via the $RUN$ process to get our desired result, disregarding the issue of divergence for a moment. The following lemma gives our first result on the relationship between Event-B refinement and CSP refinement.  

\begin{lemma} \label{lem:btref0}
If $M_0 \brefinedby M_1$ and the refinement introduces new events $N_1$ and uses the mapping $f_1$, then $f_1^{-1}(M_0)\interleave RUN_{N_1} \refinedby_{TDI} M_1$. 
\end{lemma}

\noindent {\bf Proof:}
  We assume state variables of $M_0$ and $M_1$ named as given above, i.e.\ state variables of $M_0$ are $v$ and of $M_1$ are $w$. Let $tr = \langle a_1, \ldots, a_n \rangle \in \traces(M_1)$. We need to show that $tr \in \traces(f_1^{-1}(M_0)\interleave RUN_{N_1})$. First of all note that the interleaving operator merges the traces of two processes together, i.e., the traces of $f_1^{-1}(M_0)\interleave RUN_{N_1}$ are simply those of $f_1^{-1}(M_0)$ with new events arbitrarily inserted. The proof proceeds by induction on the length of the trace.
\begin{description}
  \item[Induction base] Assume $n=0$, i.e., $tr = \langle \rangle$. By definition this means that the initialisation event $init1$ has been executed bringing the machine $M_1$ into a state $w_1$. By INV\_REF (using init as event), we find a state $v_1$ such that $J(v_1,w_1)$ and furthermore $\langle \rangle \in \traces(M_0)$ and hence also in $\traces(f_1^{-1}(M_0)\interleave RUN_{N_1})$. 
    \item[Induction step] Assume that for a trace $tr = \langle a_1, \ldots, a_{j-1} \rangle \in \traces(M_1)$ we have already shown that $tr \in traces( f_1^{-1}(M_0) \interleave RUN_{N_1})$ and this has led us to a pair of states $v_{j-1}$, $w_{j-1}$ such that $J(v_{j-1},w_{j-1})$. Now two cases need to be considered:
      \begin{enumerate}
        \item $a_j \notin N_1$: Assume $a_j$ in $M_1$ to be of the form $$\Bwhen\; H(w)\; \Bthen\; w :| BA1(w,w')\; \Bend$$ and $f_1(a_j)$ in $M_0$ of the form $$\Bwhen\; G(v)\; \Bthen\; v :| BA(v,v')\; \Bend$$ Since $a_j$ is executed in $w_{j-1}$ we have $H(w_{j-1})$. By GRD\_REF we thus get $G(v_{j-1})$. Furthermore, for $w_j $ with $BA1(w_{j-1},w_j)$ we find -- by INV\_REF --  a state $v_j$ such that $J(v_j,w_j)$ and $BA(v_{j-1},v_j)$. Hence $tr \cat \langle a_j \rangle \in \traces(f_1^{-1}(M_0)\interleave RUN_{N_1})$. 
        \item $a_j \in N_1$: Similar to the previous case. Here, $a_j$ refines skip and thus $v_j = v_{j-1}$ and the event $a_j$ is coming from $RUN_{N_1}$. 
      \end{enumerate} 
\end{description}
In the same way we can carry out a proof for infinite traces. For divergences it is even simpler as $\divergences(M_1) = \emptyset$. 
 \hfill $\Box$

\medskip 
\noindent This lemma can be generalised to a chain of refinement steps. For this, we assume that we are given a sequence of Event-B machines $M_i$ with their associated processes $P_i$, and every refinement step introduces some set of new events $N_i$. 

\begin{theorem} \label{thm:runchain}
If a sequence of processes $P_i$, mappings $f_i$, and sets  $N_i$ are such that 
\begin{eqnarray}
f_{i+1}^{-1}(P_i) \interleave RUN_{N_{i+1}} & \refinedby_{TDI} & P_{i+1} \label{refrel}
\end{eqnarray}
for each $i$, 
then
\begin{eqnarray*}
f_n^{-1}(\ldots(f_1^{-1}(P_0))\ldots ) \interleave RUN_{f_n^{-1}(\ldots f_2^{-1}(N_1)\ldots) \union \ldots \union f_{n}^{-1}(N_{n-1}) \union N_n} & \sqsubseteq_{TDI} & P_n 
\end{eqnarray*}
\end{theorem}

\noindent {\bf Proof:}
Two successive refinement steps combine to provide a relationship between $P_0$ and $P_2$ of the same form as Line~\ref{refrel} above, as follows:  
\[
\begin{array}{rcll}
f_2^{-1}(P_1) \interleave RUN_{N_2} & \sqsubseteq_{TDI} & P_2 & \mbox{(given)} \\
f_2^{-1}(f_1^{-1}(P_0) \interleave RUN_{N_1}) \interleave RUN_{N_2} & \sqsubseteq_{TDI} & P_2 & \mbox{(line (\ref{refrel}), transitivity of $\sqsubseteq$)} \\
f_2^{-1}(f_1^{-1}(P_0)) \interleave RUN_{f_2^{-1}(N_1)} \interleave RUN_{N_2} & \sqsubseteq_{TDI} & P_2 
& (\mbox{Law:\ } f^{-1}(P \interleave Q) = f^{-1}(P) \interleave f^{-1}(Q))
\\
f_2^{-1}(f_1^{-1}(P_0)) \interleave RUN_{f_2^{-1}(N_1) \union N_2} & \sqsubseteq_{TDI} & P_2 
& (\mbox{Law:\ } RUN_A \interleave RUN_B = RUN_{A \union B})
\end{array}
\]
Hence the whole chain of refinement steps can be collected together, yielding the result.
\hfill $\Box$

\subsection{Convergent and anticipated events} 

The previous result lets us relate the first and last Event-B machine in a chain of refinements. Due to the lazy abstraction operator (and the resulting possibility of defining refinement without hiding new events), we considered divergence free processes there: all processes $P_i$ representing Event-B machines, are divergence free by definition. However, Event-B refinement is concerned with a particular form of divergence and its avoidance. A sort of divergence would arise when new events (or more specifically, convergent events) could be executed forever, and this is what the proof rules for variants rule out. 

We would like to capture the impact of convergence and anticipated sets of events in the CSP semantics as well.  To do so, we first of all define the specification predicate 
\begin{eqnarray*}
CA(C,R)(u) & \defs & (\# (u \project C) = \infty \implies \# (u \project R) = \infty)
\end{eqnarray*}

\noindent Intuitively, this states that all infinite traces having infinitely many convergent ($C$) events also have infinitely many ($R$) remaining events (and thus cannot execute convergent events alone forever). In this case we say that the Event-B machine {\em does not diverge on $C$ events}. 

\begin{definition} \label{def:satca}
Let $M$ be an Event-B machine with its alphabet $\alpha M$ containing event sets $C$ and $R$ with $C \inter R = \emptyset$. $M$ {\em does not diverge on $C$ events} if $M \sat CA(C,R)$.
\end{definition}

\noindent Convergent events in Event-B machines only come into play during refinement. Thus a plain, single Event-B machine has no convergent events ($C = \emptyset$) and thus trivially satisfies the specification predicate. 

\begin{lemma} \label{lem5.4}
If $M_0 \brefinedby M_1$, and $M_1$ has convergent, anticipated, and remaining events $C_1$, $A_1$, and $R_1$ respectively, then $M_1 \sat CA(C_1,R_1)$
\end{lemma}
\smallskip 
\noindent {\bf Proof:} 
We prove this by contradiction.  Assume $\neg M_1 \sat CA(C_1,R_1)$.  Then there is some $u \in infinites(M_1)$ such that  $\#(u \project C_1) = \infty$ and $\#(u \project R_1) < \infty$.  Then there must be some $tr_0$, $u'$ such that $u = tr_0 \cat u'$ with $u' \in (C_1 \cup A_1)^\omega$ (i.e. $tr_0$ is a prefix of $u$ containing all the $R_1$ events).  Moreover, $\# u' \project C_1 = \infty$. 
 
Now since $M_0 \brefinedby M_1$ we have by GRD\_REF and INV\_REF that there is some pair of states $(v,w)$ (abstract and concrete state) reached after executing $tr_0$ for which $J(v,w)$ and $I(v)$ is true. Furthermore, $V(w)$ is a natural number. Also by $M_0 \brefinedby M_1$ we have an infinite sequence of pairs of states $(v_i,w_i)$ (for the remaining infinite trace $u'$) such that $J(v_i,w_i)$.   Since each event in $u'$ is in $A_1$ or $C_1$ we have from WFD\_REF that $V(w_{i+1}) \le V(w_i)$ for each $i$.  Further, for infinitely many $i$'s (i.e. those events in $C_1$) we have $V(w_{i+1}) < V(w_i)$.  Thus we have a sequence of values $V(w_i)$ decreasing infinitely often without ever increasing.  This contradicts the fact that the $V(w_i) \in \nat$. 
\hfill $\Box$

\smallskip
\noindent 
A number of further interesting properties can be deduced for the specification predicate $CA$. 

\begin{lemma} \label{lem:ca} Let $P$ be a CSP process and $C, C',R \subseteq \alpha P$ nonempty finite sets of events. 
\begin{enumerate}
\item If $P \sat CA(C,R)$ then $f^{-1}(P) \sat CA(f^{-1}(C),f^{-1}(R))$.
\item If $P \sat CA(C,R)$ and $N \inter C = \{ \}$ then $P \interleave RUN_N \sat CA(C,R)$.
\item If $P \sat CA(C,R)$ and $P \sat CA(C',C \union R)$ then $P \sat CA(C \union C',R)$.
\item If $P \sat CA(C,R)$ and $C \inter R = \{ \}$ then $P \hide C$ is divergence-free.
\end{enumerate}
\end{lemma}

\smallskip 
\noindent {\bf Proof:} 
\begin{enumerate}
  \item Assume that $u \in infinites(f^{-1}(P))$ and $\# (u \project f^{-1}(C)) = \infty$. From the first we get $f(u) \in infinites(P)$. From the latter it follows that $\#(f(u) \project C) = \infty$. With $P \sat CA(C,R)$ we have $\#(f(u) \project R) = \infty$ and hence $\#(u \project f^{-1}(R)) = \infty$. 
\item Let $u\in infinites(P \interleave RUN_N)$ and $\# (u \project C) = \infty$. With $N \cap C = \emptyset$ we get $\#((u \hide N) \project C) = \infty$. By definition of $\interleave$ we have $u \hide N \in infinites(P)$ ($u \hide N$ is infinite since $\#((u\hide N) \project C) = \infty$). By $P \sat CA(C,R)$ we get $\#((u \hide N) \project R) = \infty$, hence $\# (u \project R) = \infty$. 
\item Let $u \in infinites(P)$ such that $\#(u\project (C \cup C')) = \infty$. Both $C$ and $C'$ are finite sets hence either $\# (u \project C) = infty$ or $\#(u \project C') = \infty$ (or both). In the first case we get $\#(u \project R) = \infty$ by $P \sat CA(C,R)$. In the second case it follows that $\#(u \project (C \cup R)) = \infty$ and hence again $\#(u \project C) = \infty$ or directly  $\#(u \project R) = \infty$. 
\item First of all note that if $P \sat CA(C,R)$ then $P$ is divergence free. Now assume that there is a trace $tr \in divergences ( P \hide C)$. Then there exists a trace $u \in infinites(P)$ such that $tr = u \hide C$, and so $\#(u \hide C) < \infty$.   Hence $\# (u \project C) = \infty$. However, --- as $C \cap R = \emptyset$ --- $\#(u \project R) \neq \infty$ which contradicts $P \sat CA(C,R)$.  
%
\end{enumerate} 
\hfill $\Box$

\smallskip 
\noindent The most interesting of these properties is probably the last one: it relates the specification predicate to the definition of divergence freedom in CSP. In CSP, a process does not diverge on a set of events $C$ if $P \hide C$ is divergence-free. 

This gives us some results about the specification predicate for single Event-B machines and CSP processes. Next, we would like to apply this to refinements. First, we again consider just two machines. 

\begin{lemma} \label{lem:brefca}
Let $M_0 \preccurlyeq M_1$ with an associated refinement function $f_1$ and let $M_0 \sat CA(C_0,R_0)$. Then
$M_1 \sat CA(f_1^{-1}(C_0) \union C_1\;,\;f_1^{-1}(R_0))$. 
\end{lemma}

\noindent {\bf Proof:}
Assume $u \in infinites(M_1)$ and $\# (u\project (f_1^{-1}(C_0) \cup C_1) = \infty$.  
We aim to establish that $\# (u \project f_1^{-1}(R_0)) = \infty$.  We have $\#(u \project f_1^{-1}(C_0)) = \infty$ or $\# (u \project C_1) = \infty$. 

In the former case,  Lemma \ref{lem:btref0} yields that $f_1(u \project f^{-1}(\alpha M_0)) \in infinites(M_0)$.  Then
\[ \begin{array}{rl} 
\# (u \project f_1^{-1}(C_0)) = \infty & \mbox{ (given) }  \\
\# (f_1(u \project f^{-1}(C_0)) \project C_0) = \infty & \mbox{ (since renaming preserves length) } \\
\# (f_1(u \project f^{-1}(\alpha M_0)) \project C_0) = \infty & \mbox{ (since } C_0 \subseteq \alpha M_0) \\
\# (f_1(u \project f^{-1}(\alpha M_0)) \project R_0) = \infty & \mbox{ (by } M_0 \sat CA(C_0,R_0)) \\
\# (u \project f^{-1}(\alpha M_0)) \project f^{-1}(R_0) = \infty & \mbox{ (since renaming preserves length) } \\
\# (u \project f_1^{-1}(R_0)) = \infty & \mbox{ (since } R_0 \subseteq \alpha M_0) 
\end{array} \] 

%

In the latter case Lemma~\ref{lem5.4} yields that $\#(u \project R_1) = \infty$.
Then 
\[ \begin{array}{rl} 
\#(u \project R_1) = \infty & \\
\#(u \project f_1^{-1}(R_0 \cup C_0)) = \infty & \mbox{ (since } R_1 = f_1^{-1}(C_0 \cup R_0)) \\
\#(u \project f_1^{-1}(R_0)) = \infty \vee \#(u \project f_1^{-1}(C_0)) = \infty & 
\end{array} \] 
The first disjunct is the desired result, the second is the one already treated above. 

\hfill $\Box$ 


\smallskip 
\noindent Note that by Lemma \ref{lem:ca} (4) the above result implies that the machine $M_1$ does not diverge on $f_1^{-1}(C_0) \cup C_1$, in particular $M_0 \hide (f_1^{-1}(C_0) \cup C_1)$ is divergence-free.

Similar to the previous case, we can lift this to chains of refinement steps. Consider the last result with respect to two refinement steps $M_0 \brefinedby M_1 \brefinedby M_2$:  
\[ 
\begin{array}{rcll}
M_0 & \sat & CA(C_0,R_0) & \mbox{(given)} \\
f^{-1}(M_0) & \sat & CA(f^{-1}(C_0),f^{-1}(R_0)) & \mbox{(lemma~\ref{lem:ca} (1))} \\
f^{-1}(M_0) \interleave RUN_{N_1} & \sat & CA(f^{-1}(C_0),f^{-1}(R_0)) & \mbox{(lemma~\ref{lem:ca} (2),} \\
&&& \mbox{\quad since $f_1^{-1}(C_0) \inter N_1 = \emptyset$)} \\
M_1 & \sat & CA(f^{-1}(C_0),f^{-1}(R_0)) & \mbox{(lemma~\ref{lem:btref0})} \\
f_2^{-1}(M_1) & \sat & CA(f_2^{-1}(f^{-1}(C_0)),f_2^{-1}(f^{-1}(R_0))) & \mbox{(lemma~\ref{lem:ca} (1))} \\
f_2^{-1}(M_1) \interleave RUN_{N_2} & \sat & CA(f_2^{-1}(f^{-1}(C_0)),f_2^{-1}(f^{-1}(R_0))) & \mbox{(lemma~\ref{lem:ca} (2))} \\
M_2 & \sat & CA(f_2^{-1}(f^{-1}(C_0)),f_2^{-1}(f^{-1}(R_0))) & \mbox{(lemma~\ref{lem:btref0})} \\
M_2 & \sat & CA(C_2 \union f_2^{-1}(C_1) \;,\; f_2^{-1}(R_1)) & \mbox{(lemma~\ref{lem:brefca})}
\end{array}
\]
Then by applying Lemma~\ref{lem:ca}(3) to the final two lines, with $R = f_2^{-1}(f_1^{-1}(R_0))$, $C = f_2^{-1}(f_1^{-1}(C_0))$, and $C' = C_2 \union f_2^{-1}(C_1)$, we obtain
\[
\begin{array}{rcll}
M_2 & \sat & CA(C_2 \union f_2^{-1}(C_1) \union f_2^{-1}(f_1^{-1}(C_0)) \;,\; f_2^{-1}(f_1^{-1}(R_0))
\end{array}
\]
Thus if
\[ M_0 \preccurlyeq M_1 \preccurlyeq \ldots \preccurlyeq M_n \]
then collecting together all the steps yields that
\begin{eqnarray}
M_n & \sat & 
CA((f_n^{-1}(\ldots f_1^{-1}(C_0)\ldots ) \union \ldots f_n^{-1}(C_{n-1}) \union C_n) \,,\, 
f_n^{-1}(\ldots f_1^{-1}(R_0) \ldots)) \label{collect2}
\end{eqnarray}

\medskip
\noindent Finally, we would like to put together these results into one result relating the initial machine $M_0$ to the final machine $M_n$ in the refinement chain. This result should use hiding for the treatment of new events, and -- by stating the relationship between $M_0$ and $M_n \hide \{ new\ events\}$ via infinite-traces-divergences refinement -- show that Event-B refinement actually does not introduce divergences on new events. For such chains of refinement steps we always assume that $A_0 = C_0 = \emptyset$ (initially we have neither anticipated nor convergent events), and $A_n = \emptyset$ (at the end all anticipated events have become convergent). 

For this, we first of all need to find out what the ``new events'' are in the final machine. Define $g_{i,j}$ as the functional composition of the event mappings from $f_j$ to $f_i$:
\begin{eqnarray*}
g_{i,j} & = & f_i \comp f_{i+1} \comp \ldots \comp f_j
\end{eqnarray*}

\noindent Then noting the disjointness of the union, by repeated application of
\begin{eqnarray*}
C_j \uplus A_j \uplus R_j & = & 
f_j^{-1}(C_{j-1} \uplus A_{j-1} \uplus R_{j-1}) \uplus N_j 
\end{eqnarray*}
we obtain
\begin{eqnarray}
C_j \uplus A_j \uplus R_j & = & g_{1,j}^{-1}(C_0 \uplus A_0 \uplus R_0) \uplus g_{2,j}^{-1}(N_1) \uplus \ldots \uplus g_{j,j}^{-1}(N_{j-1}) \uplus N_j \label{eq:new}
\end{eqnarray}
Observe that this is a partition of $C_j \uplus A_j \uplus R_j$.
Also, by repeated application of
\begin{eqnarray*}
R_j & = & f_j^{-1}(R_{j-1})  \uplus f_j^{-1}(C_{j-1}) 
\end{eqnarray*}
we obtain
\begin{eqnarray}
R_j \uplus C_j & = & g_{1,j}^{-1}(R_0) \uplus g_{1,j}^{-1}(C_0) \uplus g_{2,j}^{-1}(C_1) \uplus \ldots \uplus g_{j,j}^{-1}(C_{j-1}) \uplus C_j \label{eq:con}
\end{eqnarray}
Observe that this is a partition of $C_j \uplus R_j$.

In a full refinement chain $M_0 \preccurlyeq \ldots \preccurlyeq M_n$ we have that $A_0 = \{ \}$,$C_0 = \{ \}$, and $A_n = \{ \}$.  Define:
\begin{eqnarray*}
NEW & = & g_{2,n}^{-1}(N_1) \uplus \ldots \uplus g_{n,n}^{-1}(N_{j-1}) \uplus N_n \\[1ex]
CON & = & g_{1,n}^{-1}(C_0) \uplus \ldots \uplus g_{n,n}^{-1}(C_{j-1}) \uplus C_n
\end{eqnarray*}
These constructions are illustrated in Figures~\ref{fig:new} and \ref{fig:con}.

\begin{figure}[t]
\begin{center}
{\tiny
\def\svgwidth{\textwidth}
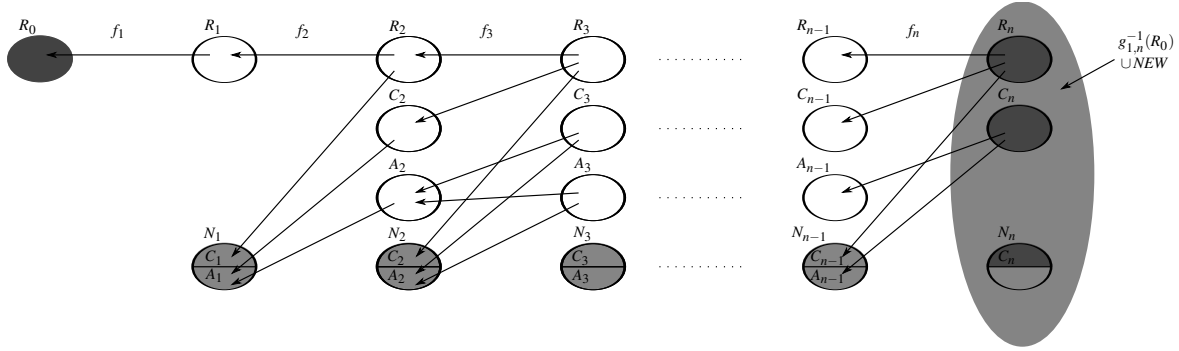
}
\end{center}
\caption{Constructing $NEW$} \label{fig:new}
\end{figure}

\begin{figure}[t]
\begin{center}
{\tiny
\def\svgwidth{\textwidth}
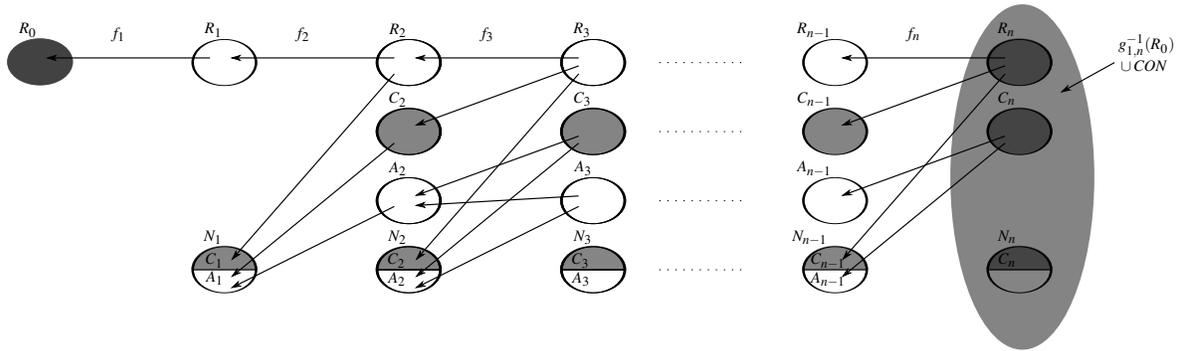
}
\end{center}
\caption{Constructing $CON$} \label{fig:con}
\end{figure}

Then from Equation~\ref{eq:new} above with $j = n$, and using $A_0 = C_0 = A_n = \{ \}$ we obtain
\begin{eqnarray*}
C_n \uplus R_n & = & g_{1,n}^{-1}(R_0) \uplus NEW
\end{eqnarray*}

\noindent From Equation~\ref{eq:con} above with $j=n$ we obtain
\begin{eqnarray*}
C_n \uplus R_n & = & g_{1,n}^{-1}(R_0) \uplus CON
\end{eqnarray*}
Hence $NEW = CON$.
From Theorem \ref{thm:runchain} and Line (\ref{collect2}) above respectively we obtain that
\begin{eqnarray*}
& f_n^{-1}(\ldots(f_1^{-1}(M_0))\ldots ) \interleave RUN_{NEW}  \sqsubseteq_{TDI}  M_n \\[2ex]
\mbox{ and } & M_n \sat CA(CON
\;,\; 
f_n^{-1}(\ldots f_1^{-1}(R_0) \ldots) \; )
\end{eqnarray*}
Lemma~\ref{lem:ca}(4) yields that $M_n \hide CON$ is divergence-free, i.e., $M_n \hide NEW$ is divergence-free. Hence by Lemma~\ref{lem:abs1} we obtain that
\begin{eqnarray} f_n^{-1}(\ldots(f_1^{-1}(M_0))\ldots ) & \sqsubseteq_{TDI} & M_n \hide NEW \label{result}\end{eqnarray}
or, equivalently, that the following theorem holds true. 

\begin{theorem} \label{thm:hiding}
Let $M_0 \preccurlyeq M_1 \preccurlyeq \ldots \preccurlyeq M_n$ be a chain of refinement steps such that $A_0 = C_0 = \emptyset$ and $A_n = \emptyset$, refining events according to functions $f_i$, and let $NEW$ be the set of events as calculated above. Then 
\[ M_0 \sqsubseteq_{TDI} f_1(f_2(\ldots f_n(M_n \hide NEW) \ldots)) \]
\end{theorem} 
\noindent {\bf Proof:}
This follows from the result in Line~\ref{result} above, using the CSP law $f(f^{-1}(P)) = P$.
\hfill $\Box$ 

\medskip
\noindent 
This result guarantees that Event-B refinement (a) does neither introduce ``new traces on old events'' nor (b) does it introduce divergences on new events. This gives us the precise account of Event-B refinement in terms of CSP which we were aiming at. 

\section{Conclusion} 

In this paper, we have given a CSP account of Event-B refinement.  The approach builds on Butler's semantics for action systems \cite{butlerphd}.  Butler's refinement rules allow new convergent events to be introduced into action systems, so that refinement steps satisfy $M_i \refinedby_{TDI} (M_{i+1} \hide N_{i+1})$, and hiding new events does not introduce divergence.  Abrial's approach to Event-B refinement generalises this approach, allowing new events to be {\em anticipated} as well as {\em convergent}, and also allowing splitting of events.  Our approach to refinement using CSP semantics reflects this generalisation and thus extends Butler's, in order to encompass these different forms of event treatment in Event-B refinement.   We do not yet handle merging events, and this is the subject of current research.  

Recently,  an Event-B$\|$CSP approach has been introduced \cite{DBLP:conf/ifm/SchneiderTW10}. It aims to combine Event-B machine descriptions with CSP \cite{schneider99} control processes, in order to support a more explicit view of control.  In this, it follows previous works on integration of formal methods \cite{butler2000,WoodcockC02,OlderogW05,SchneiderT05,Iliasov09}, which aim at complementing a state-based specification formalism with a process algebra. 

The account of refinement presented here provides the basis for a flexible refinement framework in Event-B$\|$CSP, and this is presented in \cite{stw11:techreport}.  The semantics justifies the introduction of a new status of {\em devolved}, for refinement events which are anticipated in the Event-B machine but convergent in the CSP controller.  This approach has been applied to an initial Event-B$\|$CSP case study of a Bounded Retransmission Protocol \cite{CS1104BRP}.  We aim to develop investigate further case studies.  We are in particular interested in finding out whether the work of showing divergence-freedom (and also deadlock-freedom) can be divided onto the Event-B and CSP part such that for some events convergence is guaranteed by showing the corresponding proof obligations in Event-B while for others we just look at divergence-freedom of the CSP process. The latter part could then be supported by model checking tools for CSP, like FDR \cite{fdr}. 
  
\bibliographystyle{eptcs}
\bibliography{references-eptcs}

\end{document}